%% file: GeneralEquilibrium.tex
\documentclass[12pt]{article}
\usepackage{graphicx,amsthm,amssymb,amsmath}
\usepackage{fullpage}
\usepackage{array}
\usepackage[pdftex]{hyperref}
\hypersetup{
 colorlinks=true,
 urlcolor=blue
}

\newtheorem{theorem}{Theorem} 
\newtheorem{proposition}[theorem]{Proposition} 

\newtheorem{lemma}[theorem]{Lemma}

\theoremstyle{definition}

\newtheorem{definition}{Definition}

\theoremstyle{remark}

\author{Dashiell E.A.\,Fryer \\ The University of Illinois at Urbana-Champaign}
\title{On the Existence of General Equilibrium in Finite Games and General Game Dynamics}
\date{\today}

\begin{document}
\maketitle

\begin{abstract}
A notion of incentive for agents is introduced which leads to a very general notion of an equilibrium for a finite game. Sufficient conditions for the existence of these equilibria are given. Known existence theorems are shown to be corollaries to the main theorem of this paper. Furthermore, conditions for the existence of equilibria in certain symmetric regions for games are also given. 

From the notion of general equilibrium, a general family of game dynamics are derived. This family incorporates all canonical examples of game dynamics. A proof is given for the full generality of this system. 
\end{abstract}

\input{Notation}
\input{IncentiveEquilibrium}
\input{IncentiveDynamics}

\bibliography{refs}
\bibliographystyle{amsalpha}

\end{document}

%% file: Notation.tex
\section{Notation and Definitions}
We shall denote the finite set of agents by $N = \{1,2,\ldots,n\}$ for some $n\in\mathbb{N}$. Each agent $i$ is endowed with a finite set of pure strategies, which will be denoted $S_i = \{1,2,\ldots,s_i\}$, with $s_i\in\mathbb{N}$ as well. To allow the agents to mix their strategies, they may choose strategies from the simplex on $s_i$ vertices, \[\Delta_i=\left\{x_i\in\mathbb{R}^{s_i}\left|x_{i\alpha}\geq 0, \sum_\alpha x_{i\alpha}=1\right.\right\},\] which is the convex hull of $S_i$, or equivalently the space of probability distributions over the finite set. For simplicity we will embed $S_i$ in $\Delta_i$ such that $\alpha\in S_i \mapsto e_{i\alpha}\in\mathbb{R}^{s_i}$ where $e_{ik}$ is the $k$th standard unit vector in the Euclidean space $\mathbb{R}^{s_i}$. We denote $S=\times_i S_i$ and $\Delta=\times_i \Delta_i$ as the pure and mixed strategy spaces respectively for the game. 

It is often convenient to denote the pure and mixed strategy spaces without a particular player; $S_{-i}$, and $\Delta_{-i}$ respectively. We define $S_{-i} = \times_{j\neq i} S_j$, and $\Delta_{-i}=\times_{j\neq i} \Delta_j$. Elements in these sets can be interpreted many different ways. In particular $S_{-i}$ is a $s_{-i}=\frac{|S|}{s_i}$ dimensional space and we would prefer to identify elements in this space with standard unit vectors in $\mathbb{R}^{s_{-i}}$ as before. Unfortunately, there are $s_{-i}!$ ways to accomplish this. In practice, we will only use this identification when we will sum over all possible combinations of pure strategies. Using a different identification will simply result in a permutation of terms in a finite sum, which of course has no effect. $k\in S_{-i}$ is a multi-index given by $(k_1,\ldots,k_{i-1},k_{i+1},\ldots,k_n)$. Our embedding, given by $k\in S_{-i}\mapsto e_{-i\beta}\in\mathbb{R}^{s_{-i}}$, extends to $\Delta_{-i}$ such that $x_{-i}\in\Delta_{-i} = \sum_\beta x_{-i\beta}e_{-i\beta}$ with $x_{-i\beta} = \prod_{j\neq i} x_{jk_j}$. If we have a single agent we will interpret $S_{-i}$, $\Delta_{-i}$, $s_{-i}$, and $x_{-i}$ as $S$, $\Delta$, $s$ and $x$ respectively.

We will also adopt a convention for replacement for part of a strategy profile, $x\in\Delta$. We write $(t_i,x_{-i})\in\Delta = (x_1, x_2,\ldots,t_i,\ldots,x_n)$, where the $i$th component of $x$ has been replaced by another strategy $t_i\in\Delta_i$. 

Each agent will have a utility function defined over the set of all possible combinations of pure strategies $S$. We will denote this utility \[u_i: S\rightarrow\mathbb{R}.\] These utility functions have unique $n$-linear extensions to $\Delta$ given by \[u_i(x)=\sum_\alpha x_{i\alpha}u_i(e_{i\alpha},x_{-i}) = x_i^TA_ix_{-i},\] where $a_{i\alpha\beta} = u_i(e_{i\alpha},e_{-i\beta})$ and $A_i = \{a_{i\alpha\beta}\}$. We will simply refer to these extensions as the utility functions from now on.

%% file: IncentiveEquilibrium.tex
\section{Game Equilibrium}
The Nash equilibrium~\cite{nash1950equilibrium} is ubiquitous throughout game
theory. The rise of evolutionary game theory has put new emphasis on dynamics of rationality as opposed to static equilibrium concepts. Most of these dynamic models are focused on either the Nash equilibrium itself or some refinement of the Nash equilibrium, eg. evolutionary stable strategies (ESS)~\cite{smith1974theory}, $\epsilon$-equilibrium~\cite{everett1957recursive}, etc. However, the question of applicability of the Nash equilibrium to actual human actors is still open. Often, in practice, the Nash equilibrium does not approximate actual behavior in a given game; see Prisoner's Dilemma~\cite{rapoport1965prisoner} or Traveler's Dilemma~\cite{basu1994traveler,basu2007traveler} for instance.

We open up the interpretation of an equilibrium by first generalizing the notion of incentive for an agent. In the sequel we will derive from this interpretation a family of differential equations that can account for different updating procedures used by agents. First however, we will show there exists equilibrium in games with general incentives requiring minimal conditions.

\subsection{Incentive Equilibrium}

We begin the treatment of general equilibrium by starting with Nash's second proof of existence in finite games~\cite{nash1951non}. \begin{definition}
	A strategy profile $x\in \Delta$ is a \emph{Nash equilibrium} if and only if \[ u_i(x) = \max_{t\in \Delta_i} u_i(t,x_{-i}),\ \forall i.\] 
\end{definition}
We may simplify this definition by further linearly extending the utility functions to all of $\mathbb{R}^m$, where $m = \prod_i s_i$. This results in $m$-linear functions which are harmonic on all of $\mathbb{R}^m$. We may therefore invoke the maximum principle on the closed convex space $\Delta$ recursively to deduce that the $u_i$'s are maximized (and minimized) in $S$. Therefore, \[\max_{t\in \Delta_i} u_i(t,x_{-i}) = \max_\alpha u_i(e_{i\alpha}, x_{-i}),\] and thus we can give an equivalent definition for the Nash equilibrium as follows:
\begin{definition}
	A strategy profile $x\in \Delta$ is a \emph{Nash equilibrium} if and only if \[ u_i(x) = \max_\alpha u_i(e_{i\alpha},x_{-i}),\ \forall i.\]  
\end{definition}
Thus it was natural for Nash to define a class of continuous incentive functions by \[\varphi_{i\alpha}^{Nash}(x) = (u_i(e_{i\alpha},x_{-i}) - u_i(x))_+\] where \[ (x)_+ = \max (0,x). \] It is at this point where we are ready to define the updating protocol by which agents will discreetly change their strategies. We define the map \[ T(x): \times_i\mathbb{R}^{s_i}\rightarrow \times_i\mathbb{R}^{s_i}\] where \[ T(x)_i = \frac{x_i + \sum_\alpha \varphi_{i\alpha}(x)e_{i\alpha}}{1+\sum_\beta \varphi_{i\beta}(x)}.\] It is easily verified that the sum of the coefficients of $T(x)_i$ is 1 if $x_i\in\Delta_i$, however, if $x_{i\alpha} = 0$ we must have $\varphi_{i\alpha}(x) \geq 0$ in order to preserve the simplex. We also require $\sum_\beta\varphi_{i\beta}(x)\neq -1$ for any $x\in\times_i\mathbb{R}^{s_i}$. This leads us to our definition of generalized incentive.
\begin{definition}
	A function $\varphi(x): \times_i\mathbb{R}^{s_i} \rightarrow
	\times_i\mathbb{R}^{s_i}$ is an \emph{incentive function} if and only if it satisfies both of the following conditions for all players $i$:
	\begin{enumerate}
		\item $x_{i\alpha} = 0 \Rightarrow \varphi_{i\alpha}(x) \geq 0,\ \forall \alpha$
		\item $\sum_\beta\varphi_{i\beta}(x)\neq -1$ for any $x\in\times_i\mathbb{R}^{s_i}.$
	\end{enumerate}
\end{definition}
If we have a function defined as above we may simply refer to it as the incentive for the game. 

To complement our definition of incentive we must redefine equilibrium for the
game to account for the general incentive. First, we will produce conditions for
the mapping $T(x)$ to have a fixed point.
\begin{align}
 0 & = T(x)_i - x_i,\ \forall i \\
 & = \frac{x_i + \sum_\alpha \varphi_{i\alpha}(x)e_{i\alpha}}{1+\sum_\beta \varphi_{i\beta}(x)} - x_i \\
 & = \frac{\sum_\alpha \varphi_{i\alpha}(x)e_{i\alpha} - x_i\sum_\beta \varphi_{i\beta}(x)}{1+\sum_\beta \varphi_{i\beta}(x)} \\
& \Leftrightarrow \sum_\alpha \varphi_{i\alpha}(x)e_{i\alpha} = x_i\sum_\beta \varphi_{i\beta}(x) \label{eqnref:parallel}\\
& \Leftrightarrow \varphi_{i\alpha}(x) = x_{i\alpha}\sum_\beta \varphi_{i\beta}(x),\ \forall i, \alpha
\end{align}
Note that at a fixed point, \ref{eqnref:parallel} says that $\varphi_i$ is parallel to $x_i$. Furthermore, $\varphi_{i\alpha}(x)/x_{i\alpha}$ equals the total incentive provided that $x_{i\alpha}\neq 0$. If $\varphi_{i\alpha}(x) = 0$ at a fixed point then either $x_{i\alpha}=0$ or $\sum_\beta\varphi_{i\beta}(x)=0$, but $x_{i\alpha} = 0 \Rightarrow \varphi_{i\alpha}(x)=0$.  An intuitive description of the concept is that agents will achieve equilibrium if they either have no incentive or their incentives are in line with their current strategy for the game. It is convenient to use the notation \[\dfrac{\varphi_i(x)}{x_i} = \left(\dfrac{\varphi_{i1}(x)}{x_{i1}},\ldots,\dfrac{\varphi_{is_i}(x)}{x_{is_i}}\right)\] with the convention that $\varphi_{i\alpha}(x) = 0 \Rightarrow \dfrac{\varphi_{i\alpha}(x)}{x_{i\alpha}} = 0$.
\begin{definition}
A strategy profile $\hat{x}$ is an \emph{incentive equilibrium} if and only if \[ \hat{x}_i\cdot \frac{\varphi_i(\hat{x})}{\hat{x}_i} = \max_{x_i\in \Delta_i} x_i\cdot \frac{\varphi_i(\hat{x})}{\hat{x}_i},\ \forall i. \] 
\end{definition}
Note that if we maximize the right hand side of the equation with respect to $x_i$ under the simplex constraint we must have $\varphi_{i\alpha}(\hat{x})/\hat{x}_{i\alpha}$ all equal wherever $\hat{x}_{i\alpha}\neq 0$ and 0 elsewhere. The left hand side is clearly $\sum_\beta\varphi_{i\beta}(\hat{x})$. Therefore, $T(\hat{x}) = \hat{x} \Leftrightarrow \hat{x}$ is an incentive equilibrium. 

The previous definition can be further simplified using identical arguments as in the standard Nash equilibrium definition. 
\begin{definition}
A strategy profile $\hat{x}$ is an \emph{incentive equilibrium} if and only if \[ \hat{x}_i\cdot \frac{\varphi_i(\hat{x})}{\hat{x}_i} = \max_{\alpha \in S_i} \frac{\varphi_{i\alpha}(\hat{x})}{\hat{x}_{i\alpha}},\ \forall i. \]
\end{definition}
From this definition it is clear that, as in the case of a Nash equilibrium, $\hat{x}_{i\alpha} = 0$ whenever $ \varphi_{i\alpha}(\hat{x})/\hat{x}_{i\alpha}$ is not maximal.

The following lemma will be very useful for proving not only our main theorem that an incentive equilibrium exists in every finite game, but will also allow us to identify equilibrium points in games that have certain symmetries. 
\begin{lemma}
If the incentive is continuous, a fixed point exists for $T$ in any closed convex $U \subset \times_i\mathbb{R}^{s_i}$ that is left invariant by $T$.
\end{lemma}
\begin{proof}
Given the assumptions, $T$ maps from $U$ to $U$ continuously and thus Brouwer's fixed point theorem guarantees the existence of a fixed point for $T$ in $U$.
\end{proof}

We now have all the tools necessary to prove the main theorem.
\begin{theorem}
If the incentive is continuous, an incentive equilibrium point $\hat{x}$ exists for any finite game.
\end{theorem}
\begin{proof}
We have defined the incentive functions such that the updating protocol $T(x)$ defined above is a continuous map from $\Delta$ to $\Delta$ and thus by our lemma, there exits an $\hat{x}\in\Delta$ such that $T(\hat{x}) = \hat{x}$. $T(\hat{x}) = \hat{x} \Leftrightarrow \hat{x}$ is an incentive equilibrium.
\end{proof}

Other consequences of our lemma can also be obtained quite simply. For example, suppose a two player game has the property that the incentive is continuous and $\varphi_1(x) = \varphi_2(x)$ for every $x\in\Delta$ such that $x_1=x_2$. The closed convex subset $U=\{x\in\Delta|x_1=x_2\}$ is left invariant by $T$ and thus an incentive equilibrium point exists in $U$. We can generalize this to symmetric $n$-player games. Denote the symmetric group on a finite set $X$ as $\text{Sym}(X)$.

\begin{proposition}
Suppose all players have the same pure strategy space, $S_1$. Let $U = \{x\in\Delta| x_{1\sigma_i(\alpha)} = x_{i\alpha}, \text{for some } \sigma_i\in\text{Sym}(S_1)\}$. If $\varphi(x)$ is a continuous incentive and $\varphi_{1\sigma_i(\alpha)}(x) = \varphi_{i\alpha}(x)$ for every $x\in U$, then an incentive equilibrium exists in $U$.
\end{proposition}
\begin{proof}
$U$ is closed and convex. U is left invariant by $T$ as \begin{align*}
x'_{i\alpha} & = T_{i\alpha}(x)
 = \frac{x_{i\alpha} + \varphi_{i\alpha}(x)}{1+\sum_\beta \varphi_{i\beta}(x)} \\
& = \frac{x_{1\sigma_i(\alpha)} + \varphi_{1\sigma_i(\alpha)}(x)}{1+\sum_\beta \varphi_{1\sigma_i(\beta)}(x)} = T_{1\sigma_i(\alpha)}(x) \\ & = x'_{1\sigma_i(\alpha)}
\end{align*} Thus our lemma guarantees the existence of a fixed point of $T$ in $U$ which is a subset of $\Delta$. As a fixed point in $\Delta$, it is also incentive equilibrium.
\end{proof}

\section{Examples}

We will now discuss some specific examples of incentives. 

\subsection{Canonical Examples}

We will refer to a collection of incentives that have been well studied in other venues. They are all very closely related to the Nash equilibrium. In fact they all share the property that an interior Nash equilibrium is an incentive equilibrium.

\subsubsection{Nash Incentive}

Above it was noted that Nash defined a family of functions as \[\varphi^{Nash}_{i\alpha}(x) = (u_i(e_{i\alpha},x_{-i}) -u_i(x))_+\] for every player $i$ and every strategy $\alpha\in S_i$. $\varphi^{Nash}(x)$ is trivially an incentive function as it is non-negative in every component at every $x\in\Delta$. Clearly this incentive is continuous as $f(x) = (x)_+$ and $u_i(x)$ are both continuous. Thus an incentive equilibrium exists for every finite game. 

We expect that the incentive equilibrium must in fact be a Nash equilibrium. If $\hat{x}$ is a Nash equilibrium, $\varphi(\hat{x}) = 0$ and thus every Nash equilibrium is an incentive equilibrium. Conversely, if $\hat{x}$ is an incentive equilibrium we have several possibilities. If $\sum_\beta \varphi_{i\beta}^{Nash}(\hat{x}) = 0$, $\hat{x}$ is a Nash equilibrium. It suffices then to consider the case when the sum is positive. Also we need not consider the case when $\hat{x}_{i\alpha}=0$ as this occurs if and only if $\varphi_{i\alpha}^{Nash}(\hat{x}) = 0$. Thus we can assume $\hat{x}_{i\alpha} > 0$. This can occur in an incentive equilibrium if and only if $\varphi_{i\alpha}^{Nash}(\hat{x}) >0$, which implies $u_i(e_{i\alpha},\hat{x}_i) > u_i(\hat{x})$ for any $\alpha$ such that $\hat{x}_{i\alpha}>0$. For these $\alpha$ the inequality, $\hat{x}_{i\alpha}u_i(e_{i\alpha},\hat{x}_i) > \hat{x}_{i\alpha}u_i(\hat{x})$, must also hold. If we sum over these $\alpha$ we obtain the impossible condition, $u_i(\hat{x}) > u_i(\hat{x})$. Thus at equilibrium $\varphi^{Nash}(\hat{x}) = 0$, which implies $\hat{x}$ is a Nash equilibrium.

\subsubsection{Replicator Incentive}
Interestingly, the replicator dynamics ~\cite{taylor1978evolutionary}, specifically the $n$-population models, given by $\dot{x}_{i\alpha} = x_{i\alpha}(u_i(e_{i\alpha},x_{-i})-u_i(x))$, provide incentive functions as well. Define $\varphi_{i\alpha}^{R}(x) = x_{i\alpha}(u_i(e_{i\alpha},x_{-i})-u_i(x))$. $\varphi^R(x)$ is an incentive function since $\sum_\beta \varphi_{i\beta}^R(x) = 0$ and $x_{i\alpha}=0 \Rightarrow \varphi_{i\beta}^R(x) = 0$. The replicator incentive function is not just continuous but analytic and thus easily satisfies the condition for existence of incentive equilibrium. 

The classification of these equilibrium points are quite easy given the total incentive is identically zero. We must have all $\varphi^R_{i\alpha}(\hat{x}) = 0$ if $\hat{x}$ is an incentive equilibrium. These functions are zero in three cases; $\hat{x}_{i\alpha} = 0$, $\hat{x}_{i\alpha}=1$, and $u_i(e_{i\alpha},\hat{x}_{-i}) = u_i(\hat{x})$. Thus in the interior of $\Delta$ our equilibrium is a Nash equilibrium. In fact given the last condition all Nash equilibria are replicator incentive equilibria. Finally, the first two conditions tell us that all $x\in S$ are equilibria in contrast to the Nash incentive.

We can actually use many different incentive functions to get the same behavior. The simplest of these is $\varphi_{i\alpha}^R(x) = x_{i\alpha}u_i(e_{i\alpha},x_{-i})$. Notice that the total incentive $\sum_\beta \varphi_{i\alpha}^R(x) = u_i(x)$ for every $x\in\Delta$, which could violate the second condition for an incentive function if the utility for any player is ever $-1$. However, we can translate our payoffs by arbitrary functions $g_i(x)$ in every component for each player $i$. Thus $\varphi^R_{i\alpha}(x) = x_{i\alpha}(u_i(e_{i\alpha},x_{-i})+g_i(x))$ and $\sum_\alpha \varphi^R_{i\alpha}(x) = u_i(x)+g_i(x)$. Furthermore, our equilibrium condition remains unchanged as $x_{i\alpha}(u_i(e_{i\alpha},x_{-i})+g_i(x))=x_{i\alpha}(u_i(x)+g_i(x))$ is satisfied if $x_{i\alpha}$ is $0$ or $1$, or if $u_i(e_{i\alpha},x_{-i}) = u_i(x)$. Thus in any finite game we can translate payoffs without changing equilibrium and every finite game has a translated version where our function is a valid incentive. In only slight contrast to the previous case the equilibria occurs when $\varphi_{i\alpha}^R(\hat{x}) \neq 0$. In general we can use $\varphi_{i\alpha}^R(x)=x_{i\alpha}(u_i(e_{i\alpha},x_{-i})+g_i(x))$ as our incentive function as long as $g(x)$ translates each of the minimum payoffs to any value greater than $-1$ (or maximum payoffs to values less than $-1$). 

\subsubsection{Projection Incentive}

The projection dynamic, originally introduced by Nagurney and Zhang\cite{nagurney1996projected} and presented in the style of Lahkar and Sandholm~\cite{lahkar2008projection}, is given by \[\dot{x}_i = F_i(x),\] where \begin{equation*}
F_i(x) = \begin{cases} (Ax)_i - \frac{1}{|S(Ax,x)|}\sum\limits_{j\in S(Ax,x)}(Ax)_j & \text{if}\quad i\in S(Ax,x) \\
0 & \text{otherwise},
\end{cases}
\end{equation*} for a single population. $S(Ax,x)$ is the set of all strategies in the support of $x$ as well as any collection of strategies that maximize the average presented in the first case above. $F_i(x)$ is clearly an incentive function for the game as $\sum_i F_i(x) = 0$ for every $x\in\Delta$. When $x_i=0$, $i\not\in\text{supp}(x)$ so $F_i(x) = 0$ or $(Ax)_i$ is part of a set that maximizes the average. In the second case $(Ax)_i$ itself must be maximal and thus $F_i(x)=0$. 

It is shown in Sandholm, Dokumaci, and Lahkar~\cite{sandholm2008projection} that the replicator dynamics and the projection dynamics share many features. Most important in this discussion are the interior equilibria, which they showed to be Nash equilibria just as is the case for the replicator incentive. However, in this case the discontinuity at the boundary means the main theorem does not apply.

\subsubsection{Best Reply}

The best reply incentive is quite easy to understand. It was originally introduced by Gilboa and Matsui~\cite{gilboa1991social}. Nash defined a function $B(x)$ in his original proof of existence, which is a set valued function that returns all of the pure strategy best replies to the strategy $x$. For our purposes we will use the function used by Young \cite{young2001individual}, where \[ BR_{i\alpha}(x) = \begin{cases} 1 & \text{if}\quad e_{i\alpha}\in B(x) \\ 0 & \text{otherwise}.\end{cases} \] To make this a function, there is a tiebreaker assumed so that only one pure strategy is taken to be the best reply. This function is a valid incentive since $BR_{i\alpha}(x)\geq 0$ and $\sum_\alpha BR_{i\alpha}(x) = 1$
for every $x\in\Delta$. The main theorem does not apply directly, however the incentive equilibria for this are exactly Nash equilibria. Thus the existence of a Nash equilibrium in all finite games guarantees the existence of an incentive equilibrium for the best reply incentive.

\subsubsection{Logit Incentive}

The logit incentive was originally introduced as a smoothed version of best reply by Fudenberg and Levine~\cite{fudenberg1998theory}. The incentive is defined as \[ \varphi^L_{i\alpha}(x) = \frac{\exp(\eta^{-1}u(e_{i\alpha},x_{-i}))}{\sum_\beta\exp(\eta^{-1}u(e_{i\beta},x_{-i}))}\] and is obviously a valid incentive function since $\varphi^L_{i\alpha}(x)>0$ for every $x$ and $\sum_\beta \varphi^L_{i\beta}(x) = 1$. The incentive is continuous and thus the main theorem applies. The incentive equilibria are exactly the fixed points of $\varphi^L(x)$.

It should be noted that as $\eta\rightarrow 0$ the incentive converges to the best reply incentive. On the other end of the spectrum as $\eta\rightarrow\infty$ this incentive approaches the zero incentive (see below). The variable $\eta$ is thought of as `noise', and is essentially a measure of error in logic, much like the $\epsilon$ in the $\epsilon$-equilibrium (see below). 

\subsubsection{Smith Incentive}

The Smith incentive was developed by Micheal J. Smith \cite{smith1984stability} to describe traffic flows. He suggests as a reasonable assumption that the rate which drivers swap from one route, $\beta$, to another route, $\alpha$, is given by the proportion of drivers on route $\beta$ times any additional cost of using $\beta$ over $\alpha$. Thus we can interpret this as an incentive to switch to $\alpha$ as \[ \varphi^S_{i\alpha}(x)=\sum_\gamma x_{i\gamma}(u_i(e_{i\gamma},x_{-i})-u_i(e_{i\alpha},x_{-i}))_+ \] where Smith would drop the $i$ as there is only a single population of drivers. 

The above function is always non-negative and thus is a valid incentive function. It is also obviously continuous thus the main theorem applies. Any Nash equilibrium, $x$, is an incentive equilibrium as $x_{i\alpha}>0 \Leftrightarrow u_i(e_{i\alpha},x_{-i}) = max_\gamma u_i(e_{i\gamma},x_{-i})$ and thus the terms in the above sum are all zero. The converse however is not generally true. The set of incentive equilibria in this case is called Wardrop equilibria \cite{wardrop1952some}.

\subsection{Other Examples}

We can also describe a number of non-canonical examples which may be of interest. 

\subsubsection{The Zero Incentive}

The trivial, or zero, incentive is given by $\varphi(x)=0$ for all $x\in\Delta$. The function is clearly a valid incentive and also trivially satisfies the conditions for the existence of an incentive equilibrium. This of course is not surprising as the zero incentive fixes every point in $\Delta$, and thus all points are incentive equilibria. If only one agent uses the zero incentive to update, it would appear to the opposition that the agent is choosing at random. In fact all elements of $\Delta_i$ are equally likely under this incentive. 

\subsubsection{Epsilon-Nash Incentive}

$\epsilon$-Nash equilibria was first introduced by Everett~\cite{everett1957recursive}. 
\begin{definition} For a fixed $\epsilon>0$, $x\in\Delta$ is an \emph{$\epsilon$-equilibrium} if and only if \[u_i(x) \geq u_i(t_i,x_{-i}) - \epsilon \quad \forall t_i\in\Delta_i,i.\]
\end{definition}
We can make a similar simplification to what we did for the Nash equilibrium. Instead of checking every $x_i\in\Delta_i$ it suffices to check only those strategies in $S_i$. We can therefore use the incentive function \[\varphi^\epsilon_{i\alpha}(x) = (u_i(e_{i\alpha},x_{-i})-u_i(x)-\epsilon)_+.\] This is clearly an incentive as it is always non-negative. It is also continuous which ensures the existence of incentive equilibrium. Of course, we already know that a Nash equilibrium exists in all finite games and a Nash equilibrium is an $\epsilon$-equilibrium for every $\epsilon>0$. 

There are simple examples of games that are repeated indefinitely which do not have Nash equilibria, but do still have $\epsilon$-equilibria for some $\epsilon>0$. While this is beyond the scope of this discussion it is worth mentioning. Within our scope are the finitely repeated prisoner's dilemmas. In these games it can be shown that the strategies tit-for-tat~\cite{axelrod1981evolution} and grim trigger are $\epsilon$-equilibria for some positive $\epsilon$ which depends on the payoffs of the one shot games.

\subsubsection{Simultaneous Updating}

While the notion of the $\epsilon$-equilibrium is very useful, it adds a degree of freedom to the problem of finding equilibria. $\epsilon$ would have to be fit to data in order to make the model viable and must be changed for each new game.

We draw inspiration for a new model from Brams' ``Theory of Moves" \cite{brams1994theory}. In his book he describes a solution concept for $2\times 2$ games that is based on a hypothetical negotiation. It is assumed that the agents begin this negotiation at a point in $S$, then each player does a hypothetical check on what would happen if they moved to their alternative pure strategy. They assume the other player will also switch and this alternating changing of strategy continues until a cycle is complete\footnote{This procedure always takes four moves, as the last move by the opposition returns the negotiation to its original state}. Then the agent, using backward induction on the cycle, decides whether or not to make the first move. The solutions are the collection of possible outcomes given the 4 possible starting positions, giving this the feel of a discrete dynamical system. 

We define an incentive function that takes into account the other players' possible reactions to an agent's move. We notice that if all agents are updating simultaneously then we can be anywhere in $\Delta$. Recall that all of the utility functions are maximized (and minimized) in $S$, so we will only make comparisons on the boundary. Our incentive is defined as \begin{align}
\varphi^{SU}_{i\alpha}(x) &= \sum_\gamma (a_{i\alpha\gamma}-u_i(x))_+ \\ &= \sum_\gamma (u_i(e_{i\alpha},x_{-i})-u_i(x)+ a_{i\alpha\gamma}-u_i(e_{i\alpha},x_{-i}))_+. \end{align} The function is a valid incentive since it is always non-negative, and is continuous which means an incentive equilibrium exists for all finite games. The incentive equilibria for $\varphi^{SU}(x)$ that lie in $S$ are very easily classified: they must be win-win situations in the sense that all players are achieving their maximum payoffs in the game.

We can further generalize this approach by adding a strategic component. The incentive \[ \varphi^{SSU}_{i\alpha}(x) = \sum_\gamma \delta_{i\alpha\gamma}(x)(a_{i\alpha\gamma}-u_i(x))_+, \] with $\delta_{i\alpha\gamma}(x) = 1$ if there exists at least one other player who would benefit from player $i$ changing to strategy $\alpha$, and 0 otherwise. The function $\delta(x)$ could be redefined to ensure any collection of players, including $N$, and thus could be used to describe cooperative game play. While these are valid incentives, they are not generally continuous. It could be possible to produce a smooth approximation using methods similar to those that created the logit incentive from the best reply.




\subsubsection{Altruism}

Altruism is a behavior exemplified by helping others without regard to personal benefit. To that end, an incentive can be described easily to fit this description. \[ \varphi^{A}_{i\alpha}(x) = \left(\min_{j\neq i} (u_j(e_{i\alpha},x_{-i})-u_j(x))) \right)_+ \] This is clearly a valid, continuous incentive and thus the main theorem applies. 

\subsubsection{Pareto Improvements and Coalitions}

Pareto improvements are characterized by the players in the game making changes in a direction in which all players' utilities increase. This is very similar to altruism with the main difference being now that each player must also be concerned with its own success as well as the group. Furthermore, we must note that direction should be interpreted as one of the unit vectors $\pi\in S$. We may write $\pi = (e_{i\alpha},e_{-i\beta})$, where $\beta\in S_{-i}$ or $\pi = (\alpha,\beta)$, taking full advantage of our multiindex, unit vector identification. 

For these intents, the incentive \begin{align*}
\varphi^{PI}_{i\alpha}(x) = & \sum_\beta \left( a_{i\alpha\beta}-u_i(x) \right)_+ \prod_{j\neq i}\left( a_{j\beta_j\pi_{-j}} - u_j(x)\right)_+ \\
 = & \sum\limits_{\pi|\pi_i=\alpha} \prod_{j}\left( a_{j\pi} - u_j(x)\right)_+ 
\end{align*} suits the purpose. If any one player does not gain from a move in direction $\pi$, then the product will be zero for every player in that direction. Otherwise, all players have the same incentive to change to their individual strategy that will achieve $\pi$. The existence theorem can be applied as continuity and the boundary condition are again trivially satisfied. Unfortunately, there are large collections of games with all strategies as incentive equilibrium. Zero sum games are a subset of this family as all strategies are pareto optimal in the sense that no pareto improvements can ever be made. 

This incentive will mimic the actions of the grand coalition where utilities are non-transferable. We can also define a more general incentive that can account for any additional coalition forming behavior. If we denote the power set of the finite set $N$ as $\mathcal{P}(N)$, then \[\varphi^{C}_{i\alpha}(x) = \sum_{\substack{\pi \\ \pi_i=\alpha}} \sum_{\substack{\Omega\in\mathcal{P}(N) \\ i\in\Omega}}\prod_{j\in\Omega}\left( a_{j\pi} - u_j(x)\right)_+ \] will collect any excess payoffs associated with forming any possible coalition with the other players. It retains continuity and the necessary boundary condition and thus incentive equilibrium of this sort also exist for all finite games.

\subsubsection{Margin of Victory and Rivalry}

While the topic of rational game play has been discussed at length, the more general notion of victory has largely been ignored in the literature. A reasonable interpretation for a player winning a game would be to outscore the opposition. To this end the following incentive, \[ \varphi^{MV}_{i\alpha}(x) =  \left( u_i(e_{i\alpha},x_{-i}) - \max_{j \neq i} u_j(e_{i\alpha},x_{-i})\right)_+ \] will always be positive for any alternative that makes its own payoff larger than any other agent's. Continuity is ensured, thus equilibrium of this sort will exist in any finite game. 

A similar concept is that of a rival. Presented here are two incentives based on the notions of hurting a rival and increasing the margin of victory over a specific rival, respectively.
\[ \varphi^{Rival}_{i\alpha}(x) = \left(  u_{\sigma(i)}(x) - u_{\sigma(i)}(e_{i\alpha},x_{-i}) \right)_+ \] 
\[ \varphi^{MVR}_{i\alpha}(x) = \left(  u_i(e_{i\alpha},x_{-i}) - u_{\sigma(i)}(e_{i\alpha},x_{-i}) \right)_+ \] In this context, $\sigma\in Sym(N)$ which fixes no index\footnote{It would seem akward to have a player be its own rival.}. Both are continuous, valid incentives.

%% file: IncentiveDynamics.tex
\section{Game Dynamics}
The emergence of the replicator equations of Taylor and
Jonker~\cite{taylor1978evolutionary} has created a renewed interest in dynamic
models of rational decision making. There are several examples of these sorts of models including, but not limited to, the logit equation~\cite{fudenberg1998theory}, best reply dynamics~\cite{gilboa1991social}, the Brown-von Neumann-Nash (BNN) equations~\cite{brown1959solutions}, projection dynamics~\cite{nagurney1996projected,sandholm2008projection}, Smith dynamics~\cite{smith1984stability}, and others. Sandholm~\cite{sandholm2010population} derives a family of differential equations, referred to as mean dynamics, given by \[ \dot{x}_i = \sum_{j\in S} x_j\rho_{ji}(u(x),x) - x_i\sum_{j\in S}\rho_{ij}(u(x),x)\] to describe the inflow and outflow of agents to and from a type $i$ within a single population. The $\rho_{ij}$ are supposed to represent the conditional switch rate of an agent switching from type $i$ to type $j$. If one were to specify this probability appropriately then one can recover all of the canonical dynamics listed above.

We seek a similarly flexible model but with incentive as the governing concept. We will proceed in such a way as to derive the BNN equations as introduced by Brown and von Neumann. As we have seen we can describe general equilibrium in games by way of incentive functions. We then allow agents to update their strategies via a revision protocol, $T(x)$ given by \[T_i(x) = \frac{x_i + \sum_\alpha \varphi_{i\alpha}(x)e_{i\alpha}}{1+\sum_\beta\varphi_{i\beta}(x)}.\] If we repeat this mapping we can think of it as a discrete time dynamical system defined recursively by $x^t = T(x^{t-1})$, where the superscript here is to denote the time step and not an exponent. 

\subsection{Incentive Dynamics}
Instead of working with the discrete time system above, we prefer to work with a continuous time differential equation if possible. To facilitate this endeavor we will redefine every incentive function to have a simple time dependence. That is \[ \tilde{\varphi}(x) := t\varphi(x). \] However, any change to the incentive function must also have an effect on the revision protocol, thus we write \begin{align*} x_i'=T_i(x,t)&:=\frac{x_i+\sum_\alpha{\tilde\varphi_{i\alpha}(x)e_{i\alpha}}} {1+\sum_\beta{\tilde\varphi_{i\beta}(x)}} \\
&=\frac{x_i+t\sum_\alpha{\varphi_{i\alpha}(x)e_{i\alpha}}}{1+t\sum_\beta{\varphi_{i\beta}(x)}}. \end{align*}
Furthermore, it is now possible to define the time derivative of $x_i$.
	\begin{align*}		
	\dot{x_i} &:= \lim_{t\rightarrow 0}\frac{x_i'-x_i}{t} \\
				  & = \lim_{t\rightarrow 0}\frac{x_i+t\sum_\alpha{\varphi_{i\alpha}(x)e_{i\alpha}}-x_i-t x_i\sum_\beta{\varphi_{i\beta}(x)}}{t+t^2\sum_\beta{\varphi_{i\beta}(x)}} \\
				  & = \lim_{t\rightarrow 0}\frac{\sum_\alpha{\varphi_{i\alpha}(x)e_{i\alpha}}- x_i\sum_\beta{\varphi_{i\beta}(x)}}{1+t\sum_\beta{\varphi_{i\beta}(x)}}\\
				  &= \sum_\alpha{\varphi_{i\alpha}(x)e_{i\alpha}}- x_i\sum_\beta{\varphi_{i\beta}(x)}
	\end{align*}
In individual coordinates we can write our family of differential equations as \[ \dot{x}_{i\alpha} = \varphi_{i\alpha} - x_{i\alpha}\sum_\beta\varphi_{i\beta}(x).
\] We will refer to this family of equations as incentive dynamics. 

It should be clear that fixed points of this family of differential equations are exactly incentive equilibria. As a consequence we have a number of incentive dynamics we can already describe rather easily. First we note that if we allow the incentive to be given as Nash originally conceived, $\varphi(x) = \varphi^{Nash}(x)$, then we recover the BNN equations as one would naturally expect. 

We note that in the special case when $\sum_\beta\varphi_{i\beta}(x) = 0$ for every $x\in\Delta$, the incentive dynamics reduce to simply $\dot x_{i\alpha} = \varphi_{i\alpha}(x)$. The $n$-population replicator equations are given by \[\varphi^R_{i\alpha}(x) = x_{i\alpha}(u_i(e_\alpha,x_{-i})-u_i(x))\] simply by recognizing that this incentive fits this special case. We have previously noted that there are many incentives that have the same equilibria as the replicator incentive above. These were given by $\varphi^R_{i\alpha}(x) = x_{i\alpha}(u_i(e_\alpha,x_{-i}) + g_i(x))$ where $g(x)$ is an arbitrary function from $\Delta$ to $\mathbb{R}^n$. We derive the replicator equations as follows \begin{align*}
\dot{x}_{i\alpha} & = \varphi^R_{i\alpha}(x) - x_{i\alpha}\sum_\beta\varphi^R_{i\beta}(x) \\
& = x_{i\alpha}(u_i(e_\alpha,x_{-i}) + g_i(x)) - x_{i\alpha}\sum_\beta x_{i\beta}(u_i(e_\beta,x_{-i}) + g_i(x)) \\
& = x_{i\alpha}u_i(e_\alpha,x_{-i}) + x_{i\alpha}g_i(x)-x_{i\alpha}\sum_\beta x_{i\beta}u_i(e_\beta,x_{-i}) - x_{i\alpha}\sum_\beta x_{i\beta}g_i(x) \\
& = x_{i\alpha}u_i(e_\alpha,x_{-i}) + x_{i\alpha}g_i(x)-x_{i\alpha}g_i(x)-x_{i\alpha}u_i(x) \\
& = x_{i\alpha}(u_i(e_\alpha,x_{-i})-u_i(x))
\end{align*}

Furthermore, we can recover all possible mean dynamics by defining the incentive \[ \varphi^M_i(x) = \sum_{j\in S} x_j \rho_{ji}(u(x),x).\] The probability of switching from strategy $i$ to $j$ is given by $\rho_{ij}(u(x),x)/R$, where $R$ is constant. Thus $\sum_{j\in S} \rho_{ij}(u(x),x) = R$ for any $i$. Hence $\varphi^M(x)$ induces the mean dynamics as follows,
\begin{align*}
\dot{x}_i & = \varphi^M_i(x) - x_i \sum_{j\in S} \varphi^M_j(x) \\
& = \sum_{j\in S} x_j\rho_{ji}(u(x),x) - x_i \sum_{j\in S} \sum_{i\in S} x_i \rho_{ij}(u(x),x) \\ 
& = \sum_{j\in S} x_j\rho_{ji}(u(x),x) - x_i \sum_{i\in S} \sum_{j\in S} x_i \rho_{ij}(u(x),x) \\
& = \sum_{j\in S} x_j\rho_{ji}(u(x),x) - x_i \sum_{i\in S} x_i \sum_{j\in S} \rho_{ij}(u(x),x) \\
& = \sum_{j\in S} x_j\rho_{ji}(u(x),x) - x_i \sum_{j\in S} \rho_{ij}(u(x),x). \\
\end{align*} 

\begin{table}
\begin{center}
\begin{tabular}{|c|c|p{0.12\textwidth}|} \hline
Incentive, $\varphi_{i\alpha}(x)$ & Incentive Dynamic & Name (Source) \\ \hline \hline 
$0$  & $\dot{x}_{i\alpha} = 0$ & Zero Incentive \\ \hline
$x_{i\alpha}(u_i(e_{i\alpha},x_{-i}) + g_i(x))$ & $\dot{x}_{i\alpha} = x_{i\alpha}(u_i(e_{i\alpha},x_{-i})-u_i(x))$ & Replicator \cite{taylor1978evolutionary} \\ \hline
$(u_i(e_{i\alpha},x_{-i}) -u_i(x))_+$ &
$\begin{array}{l@{}l}
\dot{x}_{i\alpha} &= u_i(e_{i\alpha},x_{-i}) - u_i(x))_+ \\ &- x_{i\alpha}\sum_\beta (u_i(e_{i\beta},x_{-i}) - u_i(x))_+
\end{array}$ & BNN \cite{brown1959solutions} \\ \hline
$BR_{i\alpha}(x)$ & $\dot{x}_{i\alpha} = BR_{i\alpha}(x) - x_{i\alpha}$ & Best Reply \cite{young2001individual} \\ \hline
$\dfrac{\exp(\eta^{-1}u(e_{i\alpha},x_{-i}))}{\sum_\beta\exp(\eta^{-1}u(e_{i\beta},x_{-i}))}$ & $\dot{x}_{i\alpha} = \dfrac{\exp(\eta^{-1}u(e_{i\alpha},x_{-i}))}{\sum_\beta\exp(\eta^{-1}u(e_{i\beta},x_{-i}))} - x_{i\alpha}$ & Logit \cite{fudenberg1998theory} \\ \hline
$\begin{array}{l@{}l} \sum_\gamma & x_{i\gamma}(u_i(e_{i\gamma},x_{-i}) \\ & -u_i(e_{i\alpha},x_{-i}))_+ \end{array}$ & $\begin{array}{l@{}l} x_{i\alpha} &= \sum_\gamma x_{i\gamma}(u_i(e_{i\gamma},x_{-i}) - u_i(e_{i\alpha},x_{-i}))_+ \\ & - x_{i\alpha} \sum_\beta\sum_\gamma  x_{i\gamma}(u_i(e_{i\gamma},x_{-i})-u_i(e_{i\beta},x_{-i}))_+ \end{array}$ & Smith \cite{smith1984stability} \\ \hline
\end{tabular}
\caption{Incentives for canonical dynamics}
\end{center}
\end{table}

\subsection{Generality}

It has been mentioned that there are other dynamical system models for game play. We would then like to know if the model presented here is in fact fully general in the sense that we can achieve all possible game dynamics with an appropriate choice of incentive. In general a game dynamic will have the form \[ \dot{x}_{i\alpha} = F_{i\alpha}(x) \] where we require $F_{i\alpha}(x)$ to preserve the simplex. Therefore, we must have $\sum_\alpha F_{i\alpha}(x) = 0$ for every $x\in\Delta$ and $i\in N$. Also, we must have $F_{i\alpha}(x) \geq 0$ if $x_{i\alpha} = 0$. These conditions make $F(x)$ an incentive function by our definition and our incentive dynamics are exactly $\dot{x}_{i\alpha} = F_{i\alpha}(x)$. Therefore, we can recover any valid game dynamic by an appropriate choice of incentive. As noted above, the incentives that generate a specific dynamic need not be unique. In fact, any valid incentive, $\varphi(x)$, has a dynamically equivalent incentive, $\tilde\varphi_{i\alpha}(x) = \varphi_{i\alpha}(x) - x_{i\alpha}\sum_\beta\varphi_{i\beta}(x)$, with the property, $\sum_\alpha \tilde\varphi_{i\alpha}(x) = 0$. This allows us to drop the restriction on incentives that $\sum_\alpha\varphi_{i\alpha}(x)\neq -1$ as behavior on the boundary is the only requirement that must be retained. 

\begin{definition}[Incentive]
A function $\varphi(x): \times_i\mathbb{R}^{s_i} \rightarrow
	\times_i\mathbb{R}^{s_i}$ is an \emph{incentive function} if and only if it satisfies $x_{i\alpha} = 0 \Rightarrow \varphi_{i\alpha}(x) \geq 0,\ \forall \alpha$.
\end{definition}